\documentclass[]{article} 
\usepackage{mathtools}
\usepackage{float}
\usepackage[utf8]{inputenc}
\usepackage{verbatim}
\usepackage{smartdiagram}
\usepackage{todonotes}
\usepackage{caption}
\usepackage{amsthm}

\usepackage{graphicx}
\usepackage{multicol}
\usepackage{footmisc}
\usepackage{cite}
\usepackage{url}

\usepackage{enumerate}

\usepackage{amsmath, amssymb, amsfonts}

\usepackage[linesnumbered, boxed, lined, ruled,vlined]{algorithm2e}

\usetikzlibrary{arrows.meta,positioning}
\usetikzlibrary{shapes.geometric}
\usepackage{tikz}
\usepackage[miktex]{gnuplottex}
\usepackage{xcolor}
\usetikzlibrary{arrows}
\usetikzlibrary{matrix}
\usetikzlibrary{decorations.pathreplacing,angles,quotes}
\usetikzlibrary{calc}
\usetikzlibrary{shapes.misc, positioning, shapes.geometric}
\usetikzlibrary{fit, graphs}
\usetikzlibrary{mindmap,shadows}

\tikzstyle{svertex}=[circle,inner sep=0.cm, minimum size=1.3mm, fill=black, draw=black]
\tikzstyle{bagvertex}=[circle, inner sep=0.cm, minimum size=6mm, fill=none, draw=black, thick]
\tikzstyle{vertex}=[circle, inner sep=0.cm, minimum size=2mm, fill=none, draw=black]
\tikzstyle{novertex}=[circle, inner sep=0.cm, minimum size=0mm, fill=none, draw=none]

\newtheorem{definition}{Definition}[section]
\newtheorem{theorem}[definition]{Theorem}
\newtheorem{lemma}[definition]{Lemma}

\newcommand{\NP}{\mathcal{NP}}
\newcommand{\N}{\mathbb{N}}

\newcommand{\R}{\mathbb{R}}
\renewcommand{\O}{\mathcal{O}}
\newcommand{\rewardsum}{\sum_{i: A_i\subseteq S} a_i}
\newcommand{\penaltysum}{\sum_{j: B_j\cap S \neq \emptyset} b_j}
\newcommand{\Penalties}{\mathcal{B}}
\newcommand{\Rewards}{\mathcal{A}}

\newcommand{\CiEntry}{C_i(S, n_1,\dots, n_h)}
\newcommand{\CjEntry}{C_j(S, n_1,\dots, n_h)}

	\title{The Reward-Penalty-Selection Problem}
\author{T. Heller\footnote{\texttt{till.heller@itwm.fraunhofer.de}, corresponding author}\and K.-H. Küfer \and S.O. Krumke}

\begin{document}
	\maketitle
	\begin{abstract}
		The \emph{Set Cover Problem (SCP)} and the \emph{Hitting Set Problem (HSP)} are well-studied optimization problems. In this paper we introduce the \emph{Reward-Penalty-Selection Problem} which can be understood as a combination of the SCP and the HSP where the objectives of both problems are contrary to each other. Applications of the RPSP can be found in the context of combinatorial exchanges in order to solve the corresponding winner determination problem. We give complexity results for the minimization and the maximization problem as well as for several variants with additional restrictions. Further, we provide an algorithm that runs in polynomial time for the special case of laminar sets and a dynamic programming approach for the case where the instance can be represented by a tree or a graph with bounded tree-width. We further present a graph theoretical generalization of this problem and results regarding its complexity.  
	\end{abstract}

\section{Introduction}
A combinatorial exchange is often used for an efficient distribution of goods among sellers and buyers, where each of them can submit their preferences by complex bids to the exchange. Application for combinatorial exchanges can be found in various settings, ranging from distributing slots for airplanes (cf. \cite{rassenti1982combinatorial, gruyer2003auctioning}) to distributing transport freight tours (cf. \cite{ewe2011, ackermann2014modeling}).

The essential part of such a combinatorial exchange is to solve the winner determination problem (WDP). Often the WDP can be formulated as a \emph{Set Cover Problem} or a \emph{Set Partition Problem}. In many auction settings suppliers of goods are also satisfied if they are paid in full but do not have to give away all the goods offered, i.e. some of the goods fall back to them. We call this property of an exchange \emph{free fall back}. For a complete real-world model that takes this property into account see for example the combinatorial freight auction of Ewe (cf. \cite{ewe2011}). On their exchange, freight carriers bid on tours they want to include in their portfolio and offer tours that do not fit into their existing vehicle plan. The author focused mainly on the profit distribution and the bidding support while solution strategies for the WDP were not discussed. 

A WDP formulation that takes into account the property of free fall back can be formulated as a combination of Set Cover and Hitting Set Problem. But first, let us recall both problems. The \emph{Set Cover Problem (SCP)} is one of the classical problems of combinatorial optimization, e.g. Karp proved the $\NP$-completeness in his seminal paper~\cite{karp1975computational}. In the SCP one is given a set of elements~$N\coloneqq \{1,\dots, n\}$ and a set of sets~$S\coloneqq \{S_i | S_i\subseteq N\}$. The task is to find a minimal sized subset of $S$ such that every element is contained in at least one of the sets. Closely related to this is the \emph{Hitting Set Problem (HSP)}, where one is given a collection~$C$ of subsets of a finite set~$S$. The task here is to find a minimal subset~$S'\subseteq S$ such that $S'$ contains at least one element from each subset in $C$. Also the HSP is $\NP$-complete, (cf. \cite{garey1979computers}). Both problems can be generalized to their weighted version, where each set or element is associated with a positive weight.

The RPSP can be seen as a combination of both, the HSP and the SCP with contrary objective functions. In the RSPS one is given a ground set of elements~$N$, a set of reward sets~$A\coloneqq \{A_i | A_i\subseteq N\}$ with corresponding rewards~$a_i\in\R$ and a set of penalty sets~$B\coloneqq \{B_i | B_i\subseteq N\}$ with corresponding penalties~$b_i\in\R$. We say a set is \emph{covered} if all elements of said set are chosen and a set is \emph{hit} if at least one element is chosen. Now the task is to find a subset of elements~$S$ such that the profit function
\begin{align}
\rewardsum - \penaltysum \label{eq: rpsp objective function}
\end{align}
is maximized (minimized). In other words, we try to find a subset of elements such that as many reward sets as possible are covered and at the same time as few penalty sets as possible are hit. Given instances of the SCP and HSP, an instance of the RPSP can be constructed by taking the sets from SCP as reward sets, the elements from the HSP as penalty sets and the elements from the SCP and the sets from the HSP as players. 

Quite a lot optimization problems which are $\NP$-complete become polynomial tractable when we restrict ourselves to instances that can be modeled as a tree, e.g. the \emph{Vertex-Cover Problem} or the \emph{Dominating Set Problem}. In many cases a dynamic programming approach can be used to obtain an exact algorithm that runs in polynomial time. Unfortunately, most instances of real-world optimization problems have no inherent structure that can be modeled by a tree. The concept of \emph{tree-width} measures how tree-like a given graph is. For this, a graph is decomposed into not necessarily disjoint sets, i.e. overlapping sets, such that the interaction map between these sets form a tree-like structure. We will use the property of bounded tree-widthness in order to solve the RPSP on graphs with a tree-like structure.

%	\subsection{Our Contribution}

In the following we discuss the maximization of the rewards and the maximization of the penalty (think of reversed roles for penalty and reward) as objective functions. There are other objective functions that might be interesting from a practical point of view, such as the maximization of the number of chosen players under a given budget restriction, but this is not addressed here. 

% In Theorem~\ref{thm: max-RPSP complexity} we show that the general setting where one tries to maximize the function~\eqref{eq: rpsp objective function} can be solved in $\O(\log R\cdot T_{MF})$, where $T_{MF}$ denotes the running time of a maximum flow algorithm. 

% Conversely, the minimization problem turns out to be $\NP$-complete, as shown in Theorem~\ref{thm: min-RPSP complexity}. Note that this problem can be seen as a maximization problem as well if one considers reversed roles for the reward and penalty sets. In Theorem~\ref{thm: laminar rpsp} we show that an instance with laminar sets can be solved in $\O(n^3)$ time. If we restrict ourselves to instances where reward sets are given as singletons and the connection graph has a bounded tree-width~$k$, we prove in Theorem~\ref{thm: bounded tree width rpsp} that one can solve this problem by a dynamic programming approach in $\O(n(2^{\O(k^3)}+2^{k+1}n^{k+1}))$ time, which is polynomial.

The rest of the paper is structured as follows. In Section~\ref{sec: prob def} we give a formal problem definition and prove complexity results for the general maximization and minimization problem. In Section~\ref{sec: prob var} several problem variants and solution strategies are presented. A generalization of the RPSP in terms of a graph theoretical problem is presented in Section~\ref{sec: rpsp: graph}. We conclude with a short outlook.

\section{Problem Definition and Complexity Results}\label{sec: prob def}
In this section we give a formal definition of the \emph{Reward Penalty Selection Problem (RPSP)} \index{Reward-Penalty-Selection Problem} and state two complexity results for the general cases. 
\begin{definition}[RPSP]
	Let $N\coloneqq\{1,\dots,n\}$ denote the set of players, $\Rewards\coloneqq \{A_1,\dots, A_h\}$ the set of reward sets~$A_i\subseteq N$ with associated reward~$a_i\in\R_+$,  and $\Penalties\coloneqq \{B_1,\dots,B_l\}$ the set of penalty sets~$B_j\subseteq N$ with associated penalty~$b_j\in\R_+$. The \emph{max-RPSP} has the objective function
	\begin{align*}
	\max_{X\subseteq N} \sum_{i: A_i\subseteq X} a_i - \sum_{j: B_j\cap X \neq \emptyset} b_j,
	\end{align*}
	whereas the \emph{min-RPSP} has the objective function
	\begin{align*}
	\min_{X\subseteq N} \sum_{i: A_i\subseteq X} a_i - \sum_{j: B_j\cap X \neq \emptyset} b_j,
	\end{align*}
\end{definition}

\begin{theorem}[max-RPSP]\label{thm: max-RPSP complexity}
	The decision problem of the max-RPSP is polynomially solvable by a minimum cut computation.
\end{theorem}
\begin{proof}
	The decision problem asks if for a given number~$\alpha$ there exists a subset~$X$ such that
	\begin{align*}
	\sum_{i: A_i\subseteq X} a_i - \sum_{j: B_j\cap X \neq \emptyset} b_j \geq \alpha.
	\end{align*}
	We define the \emph{reward-penalty-selection graph} for the decision problem as a bipartite graph~$G=(A\cup B, E)$. For each reward set in $\Rewards$ we add a node~$A_i$ to $A$ and for each penalty set in $\Penalties$ we add a node~$B_j$ to $B$. Further, we add a source node~$s$, a sink node~$t$ and an artificial node~$z$. Furthermore, let the set of reward set indices be given by $I\coloneqq \{i: A_i \in \Rewards\}$. 
	
	The source node~$s$ is adjacent to all penalty nodes~$B_j$ and the sink node~$t$ is adjacent to all reward nodes~$A_i$. The artificial node~$z$ is both connected to $s$ and $t$. A reward node~$A_i$ is adjacent to a penalty node~$B_j$ if and only if their intersection is nonempty. The capacities are defined as follows:
	\begin{align*}
	c(e) = \begin{cases} b_j & \text{ if } e=(s,B_j) \\
	a_i & \text{ if } e=(A_i, t) \\
	\alpha & \text{ if } e=(s,z) \\
	\infty & \text{ otherwise. }
	\end{cases} 
	\end{align*}
	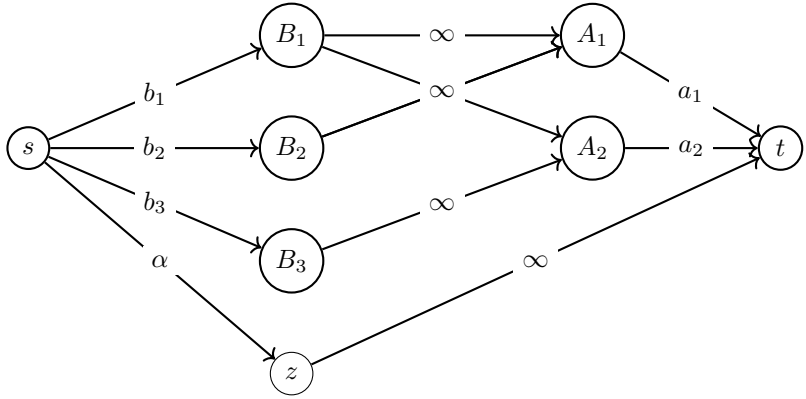
\begin{figure}[htbp!]
		\centering
		\begin{tikzpicture}
		\begin{scope}[every node/.style={circle,thick,draw}]
		\node (s) at (0,3) {$s$};
		\node (B1) at (3.5,4.5) {$B_1$};
		\node (B2) at (3.5,3) {$B_2$};
		\node (B3) at (3.5,1.5) {$B_3$};
		
		\node (A1) at (7.5,4.5) {$A_1$};
		\node (A2) at (7.5,3) {$A_2$};
		
		\node (t) at (10, 3) {$t$};
		\end{scope}
		
		\begin{scope}
		[every node/.style={circle, draw}]
		\node (z) at (3.5, 0) {$z$};
		\end{scope}
		
		\begin{scope}[
		every node/.style={fill=white,rectangle},
		every edge/.style={draw=black,thick}]
		
		\path [->] (s) edge node {$b_1$} (B1);
		\path [->] (s) edge node{$b_2$} (B2);
		\path [->] (s) edge node{$b_3$} (B3);
		\path [->] (s) edge node{$\alpha$} (z);
		
		\path [->] (B1) edge node {$\infty$} (A1);
		\path [->] (B2) edge node {$\infty$} (A1);
		\path [->] (B3) edge node {$\infty$} (A2);
		\path [->] (B1) edge node {$\infty$} (A2);
		\path [->] (B2) edge node {$\infty$} (A1);
		
		\path [->] (A1) edge node{$a_1$} (t);
		\path [->] (A2) edge node{$a_2$} (t); 
		\path [->] (z) edge node {$\infty$} (t);
		
		\end{scope}
		\end{tikzpicture}
		\caption{Example of the bipartite graph constructed in proof of Theorem~\ref{thm: max-RPSP complexity}.}
	\end{figure}
	We now show that the following two statements are equivalent:
	\begin{itemize}
		\item There exists a minimum s-t-cut~$(S,T)$ in $G$ such that for the capacity~$c(\delta^+(S)) \leq \sum_{i\in I} a_i - \alpha$ holds true.  
		\medskip
		\item There exists a set~$X$ such that $\rewardsum - \penaltysum \geq \alpha$.
	\end{itemize}
	
	\medskip
	Let $(S,T)$ denote a minimum cut with $c(\delta^+(S))\leq \sum_{i\in I} a_i - \alpha$. Then the cut capacity is given by
	\begin{align*}
	& \sum_{(s,j)\in\delta^+(S)} b_j + \sum_{(i,t)\in\delta^+(S)} a_i \leq \sum_{i\in I} a_i - \alpha \\
	\Leftrightarrow & \sum_{(s,j)\in\delta^+(S)} b_j \leq \sum_{i\in W} a_i - \alpha,
	\end{align*}
	where $W$ is defined as $\{i\in I\} \backslash \{i: (i,t)\in\delta^+(S)\}$. It holds that $N(W)\subseteq \{j: (s,j)\in\delta^+(S)\}$ since otherwise there would exist a node in $T$ and a node in $S$ that are connected which implies that the capacity of the minimum $s-t$ cut is not finite - a contradiction. Thus, we get
	\begin{align*}
	\sum_{i\in W} a_i & \geq \sum_{(s,j)\in\delta^+(S)} b_j + \alpha \\
	& = \sum_{j \in N(W)} b_j + \alpha
	\end{align*} 
	Now, $X = W \, \cap \, \{A_i : \, A_i\in\Rewards\}$ fulfills the desired condition. The converse direction follows along the same lines. Thus, the RPSP can be solved by a minimum cut computation.
\end{proof}

%In the following theorem, we will show that 
%\begin{align}
%\min + \begin{cases}
%\text{ reward when set fully chosen } \\
%\text{ penalty when at least one member is chosen } 
%\end{cases}
%\end{align}
%is $\NP$-complete. This problem is equivalent to maximizing the penalty in the above setting. By changing names (so that we maximize the reward), we get
%\begin{align}
%\max + \begin{cases}
%\text{ reward when at least one member is chosen } \\
%\text{ penalty when set fully chosen } 
%\end{cases}
%\end{align}
%is $\NP$-complete. This problem is called \emph{min-RPSP} or simply \emph{RPSP}. 
The theorem above shows that the max-RPSP can be solved in polynomial time. On the other hand, it turns out that the min-RPSP is hard to solve in general, as we see in the following theorem. From now on we will abbreviate the min-RPSP by RPSP. Note that 
\begin{align}
\min_{X\subseteq N} \sum_{i: A_i\subseteq X} a_i - \sum_{j: B_j\cap X \neq \emptyset} b_j,
\end{align}
is equivalent to
\begin{align}
\max_{X\subseteq N} \sum_{j: B_j\cap X \neq \emptyset} b_j - \sum_{i: A_i\subseteq X} a_i. \label{func: rpsp objective function}
\end{align}
Thus, we can think of the RPSP as a maximization problem. For the rest of the paper we consider the RPSP as a maximization problem with objective function~\eqref{func: rpsp objective function} and reverse the roles of reward and penalty sets, i.e. we obtain a reward from a reward set if it is hit and get a penalty from a penalty set if it is covered. 

\begin{theorem}[min-RPSP]\label{thm: min-RPSP complexity}
	The min-RPSP is $\NP$-complete, even if the penalty sets have size~$2$, the reward sets are singletons and we assume uniform rewards and penalties of $1$.
\end{theorem}
\begin{proof}
	We show the claim by a reduction from the \emph{Maximum Independent Set Problem (MIS)}. An instance of the MIS is given by an undirected graph~$G=(V,E)$ and a number~$k\in\N$.  The question posed is if there is an independent set~$S\subseteq V$, i.e. a set~$S$ such that no two nodes in~$S$ are adjacent, which has size at least~$k$. The MIS is well-known to be $\NP$-complete (cf.\cite{garey1979computers}).
	
	Given an instance of the MIS, we define the penalty and reward sets in the following way. First, we identify the set of players with the set of all nodes. For every edge~$e\in E$ we add a penalty set~$B_e$ with penalty~$1$ consisting of the two players that correspond to the incident nodes. For every node~$v\in V$ we add a reward set with reward~$1$ consisting of the player corresponding to the node. Thus, choosing both end nodes of an edge always gives a solution of strictly positive value. 
	
	Moreover, an independent set~$S\subseteq V$ in the graph~$G$ yields a solution to the RPSP of value~$|S|\geq 0$.  Thus, in particular, if there is an independent set in~$G$ of size~$k$, there is a solution to the RPSP of value~$k$.
	
	Conversely, a selection of players with profit~$k$ might also choose some penalty sets completely. In this case, by removing one chosen player that is contained in a completely chosen penalty set from the selection, the profit is decreased by one, but also increased by at least one since at least one penalty set is less chosen completely. By iterating this procedure we get a solution with profit at least~$k$ such that no penalty set is chosen completely. Now, by taking the nodes corresponding to the chosen players, we obtain an independent set of size at least~$k$.
\end{proof}

Suppose we have uniform rewards~$a$ and uniform penalties~$b$. Then the proof above shows whenever we have 
\begin{align*}
\frac{b}{a} \geq 1
\end{align*}
the min-RPSP is $\NP$-hard. In the proofs of the complexity results above we have seen how to construct a graph from an instance of the RPSP. We formalize this as follows. Each instance~$I$ of the RPSP can be represented by a \emph{connection graph} CG(I), which is a bipartite graph with node sets~$V_1\coloneqq \Rewards\cup\Penalties$, $V_2\coloneqq N$, and edges between a set node and a player node if and only if the player is contained in the considered set. The corresponding reward and penalty are associated with the respective node.

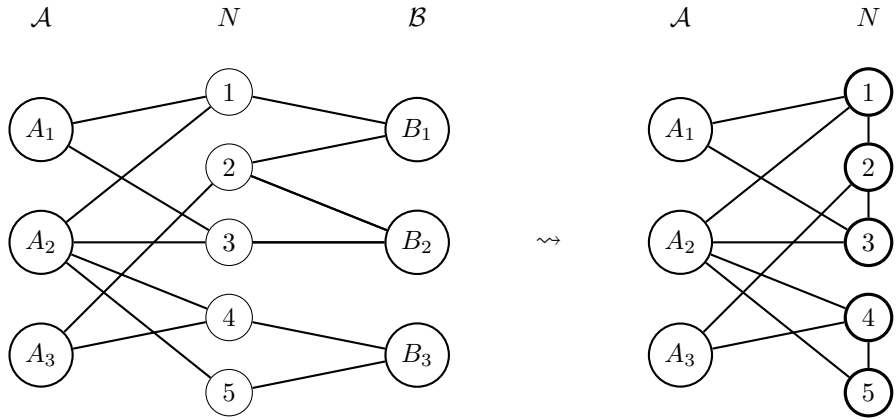
\begin{figure}[htbp!]
	\centering
	\begin{tikzpicture}[every node/.style={fill=white,rectangle}, every edge/.style={draw=black,thick}]
	\begin{scope}
	\begin{scope}[every node/.style={circle,thick,draw}]
	\draw (2.5, 6) node (A) [novertex] {$\Rewards$};
	\node (A1) at (2.5,4.5) {$A_1$};
	\node (A2) at (2.5,3) {$A_2$};
	\node (A3) at (2.5,1.5) {$A_3$};
	
	\draw (7.5,6) node (B) [novertex] {$\Penalties$}
	(7.5, 1.5) node (B3) {$B_3$};
	\node (B1) at (7.5,4.5) {$B_1$};
	\node (B2) at (7.5,3) {$B_2$};
	\end{scope}
	
	\begin{scope}
	[every node/.style={circle, draw}]
	\draw (5,6) node (N) [novertex] {$N$};
	\node (n1) at (5,5) {$1$};
	\node (n2) at (5,4) {$2$};
	\node (n3) at (5,3) {$3$};
	\node (n4) at (5,2) {$4$};
	\node (n5) at (5,1) {$5$};
	\end{scope}
	
	\begin{scope}
	
	\draw 
	(A1) edge (n1)
	(A2) edge (n1)
	(A2) edge (n3)
	(A3) edge (n2)
	(A1) edge (n3)
	(A2) edge (n5)
	(A2) edge (n4)
	(A3) edge (n4)
	(n2) edge (B1)
	(n2) edge (B2)
	(n3) edge (B2)
	(n1) edge (B1)
	(n2) edge (B2)
	(n3) edge (B2)
	(n4) edge (B3)
	(n5) edge (B3); 
	\end{scope}
	\end{scope}
	
	\begin{scope}[shift={(9.25,0)}]
	\draw (0,3) node (T) [novertex] {$\rightsquigarrow$};
	\end{scope}
	\begin{scope}[shift={(8.5,0)}]
	\begin{scope}[every node/.style={circle,thick,draw}]
	\draw (2.5, 6) node (A) [novertex] {$\Rewards$};
	\node (A1) at (2.5,4.5) {$A_1$};
	\node (A2) at (2.5,3) {$A_2$};
	\node (A3) at (2.5,1.5) {$A_3$};
	\end{scope}
	
	\begin{scope}
	[every node/.style={circle, draw, very thick}]
	\draw (5,6) node (N) [novertex] {$N$};
	\node (n1) at (5,5) {$1$};
	\node (n2) at (5,4) {$2$};
	\node (n3) at (5,3) {$3$};
	\node (n4) at (5,2) {$4$};
	\node (n5) at (5,1) {$5$};
	\end{scope}
	
	\begin{scope}
	
	\draw 
	(A1) edge (n1)
	(A2) edge (n1)
	(A2) edge (n3)
	(A3) edge (n2)
	(A1) edge (n3)
	(A2) edge (n5)
	(A2) edge (n4)
	(A3) edge (n4)
	(n1) edge (n2)
	(n3) edge (n2)
	(n4) edge (n5); 
	\end{scope}
	\end{scope}
	\end{tikzpicture}
	\caption{Example of a connection graph and the corresponding simplified connection graph.}	
\end{figure}

Note that in special cases the connection graph can be slightly adapted. Consider the case where all reward sets are given as singletons, i.e. it contains only one player and therefore is only adjacent to one player node. In this case we can identify the player nodes by the reward set nodes in order to reduce the order of the graph. This gives a bipartite graph with node set $V\coloneqq \Rewards\cup\Penalties$ which is called the \emph{reduced connection graph}. Since the reduction identifies player nodes with reward set nodes that have only one neighbour by assumption, the tree-width of the connection graph is equal to the tree-width of the reduced connection graph. 

Similar, if all penalty set nodes are of size exactly two, these degree-2-nodes can be replaced by an edge between the incident player nodes and an edge weight corresponding to the penalty. Note that this \emph{simplified graph} is used in the proof of the $\NP$-completeness of the RPSP. Again the tree-width of the connection graph is equal to the tree-width of the reduced connection graph since subdividing edges, i.e. replacing an edge by a degree-2-node that is adjacent to the incident node of the edge, does not change the size of the bags in a tree decomposition. 

Since the tree-width is preserved and together with the reduction in the proof of Theorem~\ref{thm: min-RPSP complexity} it yields that if the instance of the RPSP has singleton reward sets, penalty sets of size exactly two with an sufficiently high penalty, and the simplified connection graph has bounded tree-width, the dynamic program for the MIS problem can be used to solve it.

%	\subsection{Solving the RPSP exact}
Furthermore, suppose we have uniform rewards~$a$, uniform penalties~$b$, and let $\Delta$ denote the maximal degree of the given instance graph. If
\begin{align*}
\frac{b}{a} \leq \frac{1}{\Delta}
\end{align*}
holds, the min-RPSP is easy to solve: choosing all reward singletons is the optimal solution.

The RPSP can be formulated as an integer program. %All the following in the section about the graph generalization
\begin{align}
\textbf{(RPSP)} \qquad \quad \max \quad &\sum_{A_i\in\Rewards}a_iy_i - \sum_{B_j\in\Penalties} b_jz_j \label{ip: min-RPSP}\\
\text{s.t.}\quad & \sum_{u\in B_j} (x_u - 1) + 1 \leq z_j \qquad \qquad  \text{for all } B_j\in\Penalties \\
& y_i \leq \sum_{u\in A_i} x_u \qquad\qquad \qquad\qquad \text{for all } A_i\in\Rewards \\
& x_u,y_i,z_j \in \{0,1\}  \label{var: decision min-RPSP}
\end{align}

Clearly, an upper bound on the running time of an exact algorithm of the RPSP is given by a brute force approach which tries out all the possible selections. This runs in $\O(2^n\cdot n^2)$ since $2^n$ possible selections have to be computed and each evaluation of a given selection costs time $\O(n^2)$. In the next section more problem variants and solution strategies are discussed. 

%In order to obtain a fast exact algorithm, we want to discuss global upper and lower bounds for the optimal solution. (For good branching rules?) For this,  

%	\subsection{The min-RPSP and its LP-Relaxation}
%	We have seen above that the min-RPSP can be formulated as a maximization problem. For the rest of the chapter, we refer to the min-RPSP by RPSP if not stated otherwise. 
In the relaxation of the RPSP, the constraints~\eqref{var: decision min-RPSP} are replaced by
\begin{align}
x_u,y_i,z_j \in [0,1]  \label{var: decision relax-min-RPSP}
\end{align}

Given a solution of the relaxation, in order to obtain an integer solution we apply arithmetic rounding to every variable~$x_i$, i.e. rounding up if the decimal part is greater or equal to $0.5$ and rounding down otherwise. The rewards and penalties are then counted afterwards.

For measuring the distance of the rounded solution~$x^{\text{round}}$ to the optimal integer solution~$x^*$, we take the element-wise distance given by
\begin{align*}
\Delta_{\text{round}} \coloneqq \sum |x^*_i - x^{\text{round}}_i|.
\end{align*}

An instance configuration is a tuple~$(n,r,p,\beta)$, where $n$ denotes the number of elements, $r$ denotes the number of reward sets, $p$ denotes the number of penalty sets and $\beta$ defines the bound~$\beta n$ on the number of elements in a reward (penalty) set. For each instance configuration, we computed 1000 random instances. In Table~\ref{tab: results lp relaxation}, $\overline{\Delta}$ denotes the average distance between the rounded solution player variable values and the player variable values of the optimal integer solution. The next column shows the maximum distance between these two solutions of all of the random instances. The factor~$\overline{\alpha}$ denotes the average approximation factor where the last column denotes the worst approximation ratio over all computed instances. 

\begin{table}[htbp!]
	\centering
	\footnotesize
	\begin{tabular}{cccc|cccc} 
		n & r& p & $\beta$ & $\overline{\Delta}$ & $\Delta_{\max}$& $\overline{\alpha}$ & $\alpha_{\min}$ \\
		\hline 
		\phantom{\small{.}}&&&&&&& \\
		100 & 100 & 100 & 0.25 & 13.743 & 31 & 0.958 & 0.574 \\
		100 & 100 & 100 & 0.5 &  7.01 & 24 & 0.974 & 0.451 \\
		100 & 100 & 100 & 0.75 & 1.646 & 19 & 0.995 & 0.763 \\
		100 & 100 & 100 & 1 & 0.725 & 13 & 0.9997 & 0.913  \\
		\phantom{\small{.}}&&&&&&& \\
		100 & 150 & 50 & 1 & 0.325 & 7 & 0.9997 & 0.928 \\
		100 & 50 & 150 & 1 & 1.278 & 15 & 0.997 & 0.658 \\
	\end{tabular}
	\caption{Average distance of the optimal solution to the rounded solution.}
	\label{tab: results lp relaxation}
\end{table}

As the results show, the average approximation is fairly good for all the tested instance configurations. Despite some instances where the rounding approach yields in a bad approximation (see last column of Table~\ref{tab: results lp relaxation}), the results show that the solution obtained by the rounding procedure can be a good initial solution. In the following subsection we are interested in exact solution strategies for different problem variants. 

\section{Problem Variants}\label{sec: prob var}
In this section we consider different variants of the RPSP. Note that changing the roles of reward and penalty sets is equivalent to change the optimization direction, i.e. instead of minimizing the penalty where we get a reward by fulfilling a reward set completely and obtaining a penalty if at least one member of the set is chosen, one can also maximize the reward where a reward is obtained by choosing at least one member of a reward set and a penalty is induced by choosing a penalty set completely. In Section~\ref{sec: prob def} we have seen that the minimization problem for singleton penalty sets and reward sets of size exactly two is $\NP$-complete. 

For the rest of this section we consider a maximization problem where a reward is gained by choosing at least one member of a reward set and a penalty is induced by choosing all members of a penalty set.

\subsection{Laminar Sets}
We say a collection~$\mathcal{S}$ of sets is \emph{laminar} if and only if for two sets~$S_1,S_2\in\mathcal{S}$ either the intersection is empty or one is contained completely in the other. 
%With this, the set-inclusion~$"\subseteq"$ defines a \emph{partial order} on $\mathcal{S}$, i.e. it holds for all $S_1,S_2, S_3\in\mathcal{S}$ that 
%	\begin{itemize}
%		\item $S_1\subseteq S_1$ (reflexivity),
%		\item if $S_1\subseteq S_2$ and $S_2\subseteq S_1$, then $S_1$ is equal to $S_2$ (anti-symmetry),
%		\item if $S_1\subseteq S_2$ and $S_2\subseteq S_3$, then $S_1\subseteq S_3$ (transitivity).
%	\end{itemize} 

Now, a \emph{laminar RPSP instance} consists of a laminar collection~$\mathcal{S}\coloneqq \Rewards \cup \Penalties$. An example for a laminar RPSP instance is depicted in Figure~\ref{fig: laminar instance}. For the rest of this subsection we assume the reward (penalty) sets to be pairwise distinct, i.e. there are no two reward (penalty) sets in $\Rewards (\Penalties)$ that are equal. Note that there might be a reward set equal to a penalty set. 
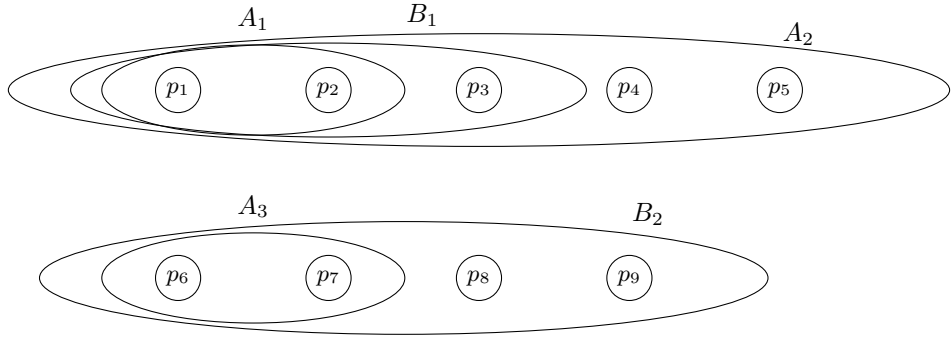
\begin{figure}
	\centering
	\begin{tikzpicture}[every edge/.style={draw=black,very thick}]
	\begin{scope}
	\draw 
	(0,0) node (p1) [vertex, minimum height=0.6cm] {\small $p_1$}
	(1, 0.95) node (A1) {$A_1$}
	(2,0) node (p2) [vertex, minimum height=0.6cm] {\small $p_2$}
	(4,0) node (p3) [vertex, minimum height=0.6cm] {\small $p_3$}
	(3.25, 1) node (B1) {$B_1$}
	(6,0) node (p4) [vertex, minimum height=0.6cm] {\small $p_4$}
	(8,0) node (p5) [vertex, minimum height=0.6cm] {\small $p_5$}
	(8.25,0.75) node (A2) {$A_2$};
	
	\node [fit=(p1) (p2), draw, ellipse, minimum height=1cm] {};
	\node [fit=(p1) (p2) (p3), draw, ellipse, minimum height=1.25cm] {};
	\node [fit=(p1) (p2) (p3) (p4) (p5), draw, ellipse, minimum height=1.5cm] {};
	\end{scope}
	
	\begin{scope}[shift={(0,-2.5)}]
	\draw 
	(0,0) node (p6) [vertex, minimum height=0.6cm] {\small$p_6$}
	(1, 0.95) node (A3) {$A_3$}
	(2,0) node (p7) [vertex, minimum height=0.6cm] {\small$p_7$}
	(4,0) node (p8) [vertex, minimum height=0.6cm] {\small$p_8$}
	(6,0) node (p9) [vertex, minimum height=0.6cm] {\small$p_9$}
	(6.25, 0.85) node (B2) {$B_2$};
	
	\node [fit=(p6) (p7), draw, ellipse, minimum height=1cm] {};
	\node [fit=(p6) (p7) (p8) (p9), draw, ellipse, minimum height=1.5cm] {};
	\end{scope}

	\end{tikzpicture}
	\caption{Example of an RPSP instance with laminar sets.}\label{fig: laminar instance}
\end{figure}

Before we continue, we recall the definition of the \emph{irreducible core of a graph} (cf. \cite{krumke2009graphentheoretische}). Let $G=(V,E)$ be a graph. We call a graph~$G'=(V,E')$ irreducible core of $G$ if the following properties hold.
\begin{itemize}
	\item $G'$ is a subgraph of $G$.
	\item Let $u,v\in V$. There exists a path from $u$ to $v$ in $G$ if and only if there exists a path from $u$ to $v$ in $G'$. 
	\item Let $G''$ be a subgraph of $G'$ with $G''\neq G'$. Then, there exist at least two nodes~$u,v$ such that there exists a $u-v-$path in $G'$ but not in $G''$.
\end{itemize}

First, we construct a graph~$G=(V, E)$ for a laminar RPSP instance. For each set~$S\in\mathcal{S}$ we add a node~$u_S$ to $V'$. We add a directed edge~$(u_{S_1}, u_{S_2})$ if $S_2\subseteq S_1$. Now the \emph{tree representation}~$T=(V, E')$ of the collection~$S$ is defined as the irreducible core of $G$. An example of the irreducible core can be found in Figure~\ref{fig: irred core}. We know that the irreducible core can be computed in $\O(m(n + m))$ (cf. Algorithm~5.3, \cite{krumke2009graphentheoretische}), where $m$ denotes the number of edges of $G$ whereas $n$ denotes the number of nodes of $G$. 

\begin{figure}
	\centering
	\begin{tikzpicture}[every node/.style={fill=white,rectangle}, every edge/.style={draw=black,very thick}]
	\begin{scope}
	\draw 
	(-0.5, 1) node (G) {$G$}
	(0.5, 0) node[vertex, minimum height=0.9cm] (A1) {$u_{A_1}$}
	(-1, -2) node[vertex, minimum height=0.9cm] (B1) {$u_{B_1}$}
	(2,-2) node[vertex, minimum height=0.9cm] (A2) {$u_{B_2}$}
	(-1, -4) node[vertex, minimum height=0.9cm] (A3) {$u_{A_3}$}
	(1, -4) node[vertex, minimum height=0.9cm] (B2) {$u_{A_2}$}
	(3, -4) node[vertex, minimum height=0.9cm] (A4) {$u_{A_4}$};
	
	\path[->] (A1) edge (B1);
	\path[->] (A1) edge (A2);
	\path[->] (B1) edge (A3);
	\path[->] (A2) edge (A4);
	\path[->] (A2) edge (B2);
	\path[->] (A1) edge (A3);
	\path[->] (A1) edge[bend left] (A4);
	\path[->] (A1) edge (B2);
	\end{scope}
	
	\begin{scope}[shift={(7,0)}]
	\draw 
	(-0.5, 1) node (T) {$T$}
	(0.5, 0) node[vertex, minimum height=0.9cm] (A1) {$u_{A_1}$}
	(-1, -2) node[vertex, minimum height=0.9cm] (B1) {$u_{B_1}$}
	(2,-2) node[vertex, minimum height=0.9cm] (A2) {$u_{B_2}$}
	(-1, -4) node[vertex, minimum height=0.9cm] (A3) {$u_{A_3}$}
	(1, -4) node[vertex, minimum height=0.9cm] (B2) {$u_{A_2}$}
	(3, -4) node[vertex, minimum height=0.9cm] (A4) {$u_{A_4}$};
	
	\path[->] (A1) edge (B1);
	\path[->] (A1) edge (A2);
	\path[->] (B1) edge (A3);
	\path[->] (A2) edge (A4);
	\path[->] (A2) edge (B2);
	\end{scope}
	\end{tikzpicture}
	\caption{Example of a graph~$G$ and its tree representation~$T$.}\label{fig: irred core}
\end{figure}
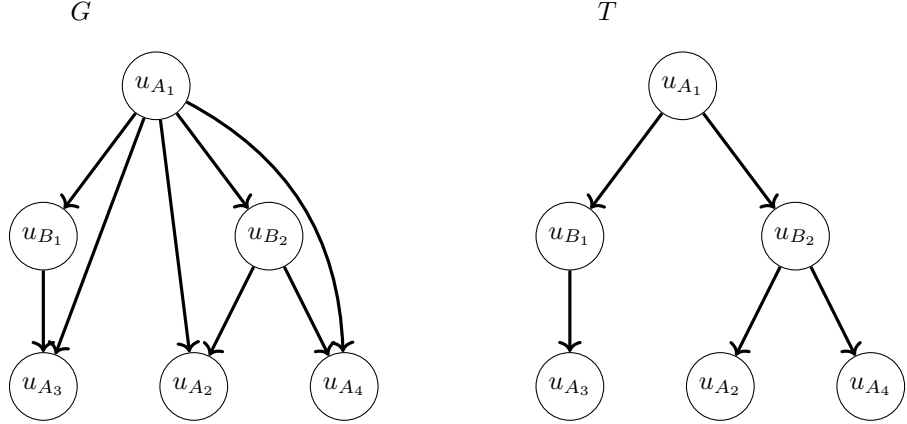 

A tree representation~$T$ such that all leaves correspond to singleton reward sets is called \emph{nice tree representation}. The following lemmas show that one can construct a nice tree representation by applying some post-processing steps to a given tree representation. 

\begin{lemma}
	Let $T$ be a tree representation and $v$ its root node. If a leaf node~$u_A$ of $T$ corresponds to a reward set~$A$ consisting of more than one element, one can contract the unique path~$P\coloneqq(u_A,\dots, v)$ and replace it by a node~$u'$ which corresponds to a singleton reward set~$A'$ with reward~$\sum_{i: u_{A_i}\in P} a_i$ and is adjacent to $v$. 
\end{lemma}
\begin{proof}
	Suppose now the new reward set~$A'$ has reward~$\sum_{i: u_{A_i}\in P} a_i$ and consists of only one element. Furthermore, it is adjacent to the root node~$v$. We further contract the path~$(u_A,\dots, v)$ and denote the obtained graph by $T'$. 
	
	First, we consider the graph~$T$. Choosing one element~$x$ from $A$ gives a reward of $\sum_{i: u_{A_i}\in P} a_i$ since $x$ is contained in every reward set~$A_i$ for which $u_{A_i}$ is contained in $P$ by construction of the tree representation~$T$. Furthermore, since $A$ contained at least two elements, no penalty set which corresponds to a node in $P$ can be covered by choosing only $x$. An optimal solution would therefore never choose more than one element of $A$ and all penalty sets containing $A$ can be neglected. In contrast, in the new graph~$T'$, if we choose the element contained in $A'$, we obtain a profit~$\sum_{i: u_{A_i}\in P} a_i$. Since the path~$P$ in $T$ is contracted, no reward is counted twice.
\end{proof}

\begin{lemma}
	Let $T$ be a tree representation and $v$ its root node. Given a path~$(u_1,u_2,\dots, v)$ from a leaf~$u_1$ to the root~$v$. We can assume that at least one of the corresponding sets~$S_1, S_2\subseteq \mathcal{S}$ is a reward set. 
\end{lemma}
\begin{proof}
	Suppose we are given a path~$(u_1,u_2,\dots, v)$ from a leaf~$u_1$ to the root~$v$ where $u_1, u_2$ correspond to penalty sets~$B_1, B_2$. If $B_2$ contains at least one element that is not contained in $B_1$, then we can delete $u_1$ from $T$ since an optimal solution would never choose an element from $B_1\backslash B_2$. If $B_2$ is equal to $B_1$, then the edge~$(u_1, u_2)$ can be contracted and we add the penalty of $B_1$ to the penalty of $B_2$. 
\end{proof}

\begin{lemma}
	Let $T$ be a tree representation and $v$ its root node. If a leaf node~$u_1$ of $T$ corresponds to a penalty set~$B$, without changing the optimal solution value, one can either delete $u_1$ from $T$ or change the node~$u_1$ by its parent node~$u_2$ which corresponds according to the observation above to a reward set~$A$ with $A=B$.
\end{lemma}
\begin{proof}
	Let $(u_1, u_2, \dots, v)$ denote the path from $u_1$ to the root~$v$. Because of the observation above we know that $u_2$ corresponds to a reward set~$A$. First, if there exists an element~$x$ in $A$ that is not contained in $B$, no element from $B$ will be chosen in an optimal solution since $x$ will hit all reward sets on the path. Hence, removing $u_1$ from $T$ does not change the optimal solution value. If $A$ is equal to $B$, we can simply change the position of $u_1$ and $u_2$ without violating the definition of $T$. Thus, in both cases we obtain a leaf which corresponds to a reward set. 
\end{proof}

Thus, without loss of generalization, one can assume to be $T$ a nice tree representation. Given a nice tree representation~$T$ of the laminar RPSP, we construct a circulation network graph~$C(T)$ by adding a source~$s$ and a sink~$t$ to the node set. Furthermore, we add edges between the source~$s$ and each leaf~$n$ of $T$ with capacity~1 and profit~0, an edge between the root node~$v$ and the sink~$t$ with $\infty$ capacity and profit~0, and an edge between $t$ and $s$ with $\infty$ capacity and profit~0. In addition, we introduce edges between two nodes~$u_1, u_2$ from $V(T)$ as follows.
If $u_2$ is the child of $u_1$ and $u_1$ corresponds to a reward set~$A$ with reward~$a$, we add two parallel edges~$e_1, e_2 = (u_1, u_2)$ with capacity~$c(e_1) = 1, c(e_2) = \infty$ and profit~$p(e_1) = a, p(e_2) = 0$. If $u_2$ is the child of $u_1$ and $u_1$ corresponds to a penalty set~$B$ with reward~$b$, we add two parallel edges~$e_1, e_2 = (u_1, u_2)$ with capacity~$c(e_1) = |B| - 1, c(e_2) = 1$ and profit~$p(e_1) = 0, p(e_2) = -b$. An example of the network graph~$C(T)$ is depicted in Figure~\ref{fig: circ graph}, where the first entry of an edge label denotes its capacity whereas the second entry denotes its profit per unit of flow.

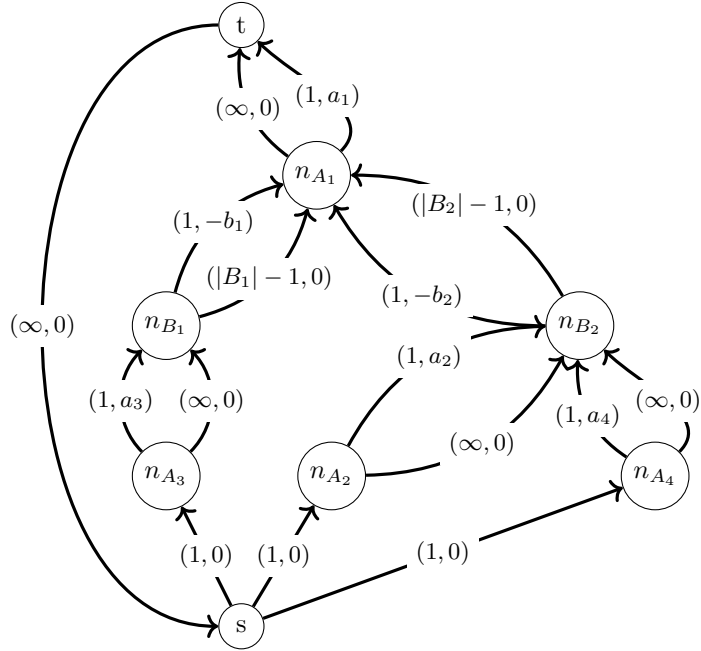
\begin{figure}
	\centering
	\begin{tikzpicture}[every node/.style={fill=white,rectangle}, every edge/.style={draw=black,very thick}]
	\begin{scope}
	\draw 
	(0, 2) node[vertex, minimum height=0.6cm] (t) {t}
	(1, 0) node[vertex, minimum height=0.9cm] (A1) {$n_{A_1}$}
	(-1, -2) node[vertex, minimum height=0.9cm] (B1) {$n_{B_1}$}
	(4.5,-2) node[vertex, minimum height=0.9cm] (A2) {$n_{B_2}$}
	(-1, -4) node[vertex, minimum height=0.9cm] (A3) {$n_{A_3}$}
	(1.2, -4) node[vertex, minimum height=0.9cm] (B2) {$n_{A_2}$}
	(5.5, -4) node[vertex, minimum height=0.9cm] (A4) {$n_{A_4}$}
	(0, -6) node[vertex, minimum height=0.6cm] (s) {s};
	
	\path[<-] (A1) edge[bend left] node {\small $(|B_1|-1,0)$} (B1);
	\path[<-] (A1) edge[bend left] node {\small $(|B_2|-1,0)$} (A2);
	\path[<-] (B1) edge[out=225, in=135] node {\small $(1,a_3)$} (A3);
	\path[<-] (A2) edge[out=315, in=45] node {\small $(\infty,0)$} (A4);
	\path[<-] (A2) edge[bend left] node {\small $(\infty,0)$} (B2);
	\path[<-] (A1) edge[bend right] node {\small $(1,-b_1)$} (B1);
	\path[<-] (A1) edge[bend right] node {\small $(1,-b_2)$} (A2);
	\path[<-] (B1) edge[out=315, in=45] node {\small $(\infty,0)$} (A3);
	\path[<-] (A2) edge[bend right] node {\small $(1,a_4)$} (A4);
	\path[<-] (A2) edge[bend right] node {\small $(1,a_2)$} (B2);
	
	\path[->] (s) edge node {\small $(1,0)$} (A3);
	\path[->] (s) edge node {\small $(1,0)$} (B2);
	\path[->] (s) edge node {\small $(1,0)$} (A4);
	
	\path[->] (A1) edge[in=315, out=45] node {\small $(1,a_1)$} (t);
	\path[->] (A1) edge[bend left] node {\small $(\infty,0)$} (t);
	
	\path[->] (t) edge[out=180, in=180] node {\small $(\infty, 0)$} (s);
	
	\end{scope}
	\end{tikzpicture}
	\caption{Example of the circulation network graph.}\label{fig: circ graph}
\end{figure}

\begin{theorem}\label{thm: laminar rpsp}
	The laminar RPSP can be solved in time~$\O(m(n+m) + T_{MF}(n,m))$. 
\end{theorem}
\begin{proof}
	Let $C(T)$ denote the circulation network of a tree representation~$T$ of a laminar RPSP instance. 
	
	Let a maximum profit circulation~$f$ in $C(T)$ with profit~$p$ be given. For a node corresponding to a reward set~$A$, the outgoing edge with capacity~1 and profit~$a$ is always chosen over the parallel edge with profit~0 by $f$. For a node corresponding to a penalty set~$B$, the circulation~$f$ tries to send as much unit of flow over the edge with profit~0. For the profit of $f$ we count exactly the reward of all reward sets~$A$ for which there exists a path~$(u_A,\dots, v)$ from a leaf node~$u_A$ to the root node~$v$ with $f(\delta^-(u_A))=1$ and the penalty of all penalty sets~$B$ for which $f(\delta^-(u_B))=|B|$. Moreover, a reward of a reward set is counted at most once since the capacity allows to send at most one unit of flow over an edge with positive profit. Now we construct a solution~$S$ to the laminar RPSP with the same profit by choosing one element in each reward set for which the corresponding node has an ingoing edge with one unit of flow.
	
	Conversely, let $S$ be a solution of the laminar RPSP with profit~$p$. Now we are sending one unit of flow to every leaf node that corresponds to a reward set which contains an element in $S$. By similar arguments as above we obtain that a circulation~$f$ with profit~$p$ can be constructed.	
\end{proof}

\subsection{Reward Singletons, Bounded Tree-Width Reduced Connection Graph}
In this subsection we consider an instance~$I$ of the RPSP with reward set singletons and a connection graph~$G$ of bounded tree-width. First, if all penalty sets of $I$ are exactly of size two, we can construct the simplified connection graph~$SCG(I)$. If now the penalty of each penalty set is larger than the sum of all rewards of the elements contained in it, the RPSP becomes exactly the \emph{MIS}, which is shown in the proof of Theorem~\ref{thm: min-RPSP complexity}. 
%Note that this further implies that there does not exist a constant factor approximation that runs in polynomial time for the RPSP, see e.g. \cite{garey1979computers}. %TODO is this in garey johnson?. 
Since the \emph{MIS} is solvable in polynomial time for an instance with bounded tree-width, the RPSP can also be solved in polynomial time by using a dynamic programming approach. 

%\section{Some Notes on Tree-width}\label{sec: preliminaries}
%In this section we briefly recall the preliminaries. 
%	\subsection{The Set Cover Problem and the Hitting S}\todo{This needs all to be done.}
%	The \emph{Set Cover Problem} can be formulated as an integer program.
%	\begin{align*}
%	\min 
%	\end{align*}
%	It is known that the decision problem of the SCP is $\NP$-complete and therefore the optimization problem is $\NP$-hard. 
%	\subsection{The Hitting Set Problem}
%	The \emph{Hitting Set Problem} can be formulated as an integer program.
%	\begin{align*}
%	\min
%	\end{align*}
%	\subsection{Tree Width}

In the following, we generalize the problem in two directions -- penalties and rewards are arbitrary as well as the size of the penalty sets. Since we consider the reward sets to be singletons, we can define the reduced connection graph, which we assume to have a bounded tree-width since this is essential for using our dynamic programming approach. 

For the rest of this subsection, we consider an instance~$I$ with reward set singletons and arbitrary penalty set size. Also, the number of reward and penalty sets is not bounded by a constant. We define a profit measure function~$p: V\mapsto \R$ that maps a reward set node to its associated reward and a penalty set node to its associated penalty. We extend this function naturally by mapping a set of nodes to the sum of the profits of each node in the set.

Let $RCG(I)$ denote the reduced connection graph for the given instance and $D=(S,T)$ its nice tree decomposition with width($D$) = $k$. For each node~$i\in V(T)$ of the tree~$T$ we define a corresponding induced subgraph~$G_i = G[V_i]$ where $V_i$ is the union of the bag~$X_i$ and all bags~$X_j$ with $j$ being a descendant of~$i$ in the tree~$T$. Now, by starting from the leaves of $T$, we compute the values~$C_i$ for every node~$i\in V(T)$, where $C_i$ contains an entry for every combination of a subset~$S\subseteq X_i$ and of a possible degree of a penalty set node for each of the $h$ penalty set nodes in the bag~$X_i$. The values of the entries are given by the maximum profit selection in the corresponding induced subgraph~$G_i$ --- more formally:
\begin{align}
\CiEntry \coloneqq \max \{&p(M) | M \text{ is a selection of reward nodes in } G_i  \notag\\
& \text{ s.t. } \quad M\cap X_i = S \label{cons: sel tree dec 1}\\
& |M\cap P_l| = n_l \quad \text{ for all penalty set nodes } P_l \label{cons: sel tree dec 2} \},
\end{align}
where $P_l$ denotes a penalty set node in $X_i$. Since we assumed that $G$ has bounded tree-width~$k$, at most $k$ penalty sets can be contained in one bag. If there exists no selection that fulfills the constraints~\eqref{cons: sel tree dec 1}, \eqref{cons: sel tree dec 2} for a given subset~$S$ and given degrees~$n_1,\dots,n_h$, we define $C_i(S, n_1, \dots, n_h) \coloneqq -\infty$.

\subsubsection{Leaves}
Let~$X_i$ be a leaf of the tree, i.e. $X_i=\{v\}$ since we have a nice tree decomposition. We can compute the entries of $C_i$ as follows. If $v$ is a reward node, we get $C_i(\{v\}) = r_v$. If $v$ is a penalty set, then we get $C_i(\{v\}, 0) = 0$ and set $C_i(\{v\}, 1) = -\infty$ since this case is not possible. Further, we set $C_i(\emptyset) = 0$.

Thus, the entries of a leaf node can be computed in $\O(1)$. 

\subsubsection{Forget Nodes}
For a forget node~$X_i$ we have $X_i = X_j \backslash \{v\}$ for a node~$v$ (see Figure~\ref{fig: schematic forget node} for a schematic representation). Note that by the definition of a tree decomposition, all neighbors of~$v$ are contained in~$G_i$. In particular, for each neighbour~$z$ of $v$ there exists at least one descendant of~$X_i$ that contains both $v$ and $z$. That being said, if a penalty set node is dropped once, it cannot be contained in any of the following bags. 

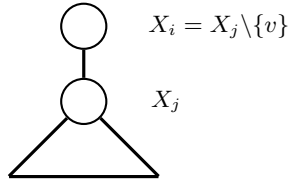
\begin{figure}[hbtp!]				
	\centering		
	\begin{tikzpicture}[every node/.style={fill=white,rectangle}, every edge/.style={draw=black,very thick}]
	\begin{scope}
	\draw
	(0,0) node (1) [bagvertex] {}
	(1.8,0) node (1a) [novertex] {\small{$X_i = X_j \backslash \{v\}$}}
	(0,-1) node (2) [bagvertex] {}
	(1.1,-1) node (2a) [novertex] {\small{$X_j$}}
	(-1,-2) node (3) [novertex] {}
	(1,-2) node (4) [novertex] {};
	
	\path[-] (1) edge (2);
	\path[-] (2) edge (3);
	\path[-] (2) edge (4);
	\path[-] (3) edge (4);
	\end{scope}
	\end{tikzpicture}
	\caption{Schematic representation of a forget node~$X_i$.}\label{fig: schematic forget node}
\end{figure}

\begin{lemma}
	It holds:
	\begin{align*}
	\CiEntry = \max \{ &C^*, C_j(S\cup\{v\}, n_1', \dots, n_h') \},
	\end{align*}
	where 
	\begin{align*}
	C^* \coloneqq \begin{cases}
	\CjEntry - p(v) &  \text{ if $N(v)\subseteq S$,} \\
	\CjEntry & \text{ otherwise,}
	\end{cases}
	\end{align*}
	and 
	\begin{align*}
	n_l' \coloneqq \begin{cases}
	n_l & \text{ if $v$ is not adjacent to $P_l$,}\\
	n_l+1 & \text{ otherwise.}
	\end{cases}
	\end{align*}
	%$n_l'$ is equal to $n_l$ if the node~$v$ is not adjacent to $P_l$ and equal to $n_l+1$ if $v$ and $P_l$ are adjacent.
	%\todo{Das muss die Aussage noch abgeändert werden, da der Fall von einem neuen PenaltyNode noch nicht berücksíchtigt ist.} DONE?
\end{lemma}
\begin{proof}
	First, suppose $v$ is a reward set node. In this case, we distinguish between two subcases. 
	
	For this let $v$ be part of the selection~$M$. Then we get $M\cap X_j = S \cup \{v\}$ and, thus,
	\begin{align*}
	\CiEntry = C_j(S\cup \{v\}, n_1',\dots, n_h'),
	\end{align*}
	with the definition of $n_l'$ above.
	
	If $v$ is not contained in the selection~$M$, then we get $M\cap X_i = M\cap X_j = S$ and, thus, 
	\begin{align*}
	\CiEntry = \CjEntry.
	\end{align*}
	
	Now suppose $v$ is a penalty set node. Clearly, $v$ cannot be contained in the selection~$M$ since it is a penalty set node.
	
	We claim that if a selection~$M$ in $G_i$ chooses all neighbors of the introduced node~$v$, then all the neighbors are contained in the current bag~$X_i$. In order to prove this, suppose there exists a neighbor~$x$ of $v$ that is chosen in the selection~$M$ but not in the current bag~$X_i$. Since $x$ and $v$ are adjacent, there exists a bag~$X_t$ in the nice tree decomposition that contains both. Since by assumption $x$ is not in $X_i$, the bag~$X_t$ is either a descendant or an ancestor of $X_i$. Since in the selection~$M$ all neighbors of $v$ are chosen, $x$ must be contained in some descendant~$X_s$ of $X_i$. Now suppose $X_t$ is an ancestor of $X_i$. Then, by the path condition of the tree decomposition, $x$ must be contained in every bag on the path from $X_s$ to $X_t$ --- a contradiction since $x$ is assumed to be not contained in $X_i$. For the other direction, suppose $X_t$ is a descendant of $X_i$. This is not possible since $X_i$ is the first bag in which the node $v$ is introduced. 
	
	Thus, since the nodes contained in $S$ have to be chosen by the selection, the question if $v$ induces a new penalty depends solely on the current subset~$S$. We get
	\begin{align*}
	\CiEntry = \begin{cases}
	\CjEntry - p(v) &  \text{ if $N(v)\subseteq S$,} \\
	\CjEntry & \text{ otherwise.}
	\end{cases}
	\end{align*}
\end{proof}
The number of subsets of a bag is given by $2^{k+1}$, since at most~$k$ elements are contained in a bag and the number of possible degree combinations is given by $n^{k+1}$, since one penalty set can have at most~$n$ neighbors in the connection graph and at most~$k+1$ penalty sets are contained in a bag. Thus, the running time of computing all values in a forget node is 
\begin{align*}
\O(2^{k+1}\cdot n^{k+1}).
\end{align*}

\subsubsection{Introduce Nodes}
Let $X_i$ be an introduce node, i.e. $X_i = X_j \cup \{v\}$ for a node~$v\in V(G)$ and the child node~$j$ in the tree decomposition (see Figure~\ref{fig: schematic introduce node}).
\begin{figure}[hbtp!]				
	\centering		
	\begin{tikzpicture}[every node/.style={fill=white,rectangle}, every edge/.style={draw=black,very thick}]
	\begin{scope}
	\draw
	(0,0) node (1) [bagvertex] {}
	(1.8,0) node (1a) [novertex] {\small{$X_i = X_j \cup \{v\}$}}
	(0,-1) node (2) [bagvertex] {}
	(1.1,-1) node (2a) [novertex] {\small{$X_j$}}
	(-1,-2) node (3) [novertex] {}
	(1,-2) node (4) [novertex] {};
	
	\path[-] (1) edge (2);
	\path[-] (2) edge (3);
	\path[-] (2) edge (4);
	\path[-] (3) edge (4);
	\end{scope}
	\end{tikzpicture}
	\caption{Schematic representation of an introduce node~$X_i$.}\label{fig: schematic introduce node}
\end{figure}
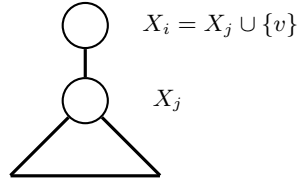
Suppose first, the introduced node~$v$ is a penalty set node. Then, for every considered subset~$S$ of the current bag~$X_i$, the node~$v$ cannot be contained in $S$ since $M$ is a subset of reward set nodes and $S=M\cap X_i$ holds true. 

Suppose now the introduced node~$v$ is a reward set node. We distinguish two different cases, either $v$ is contained in $S$ or not.

First, suppose $S$ does not contain $v$. We claim that 
\begin{align*}
\CiEntry = \CjEntry.
\end{align*}

Since $M\cap X_i = S = M\cap X_j$ holds, the optimal selection~$M$ in $G_j$ is a feasible selection in $G_i$. Thus, we get 
\begin{align*}
\CiEntry \geq \CjEntry.
\end{align*}
For the other inequality, take an optimal selection~$M_i$ in $G_i$. Since $v$ is not in $S$, we know that $M_i\cap X_i = S = M_i\cap X_j$ holds true. Therefore $M_i$ is also a feasible selection in $G_j$. Thus, we get 
\begin{align*}
\CiEntry \leq \CjEntry.
\end{align*}

Second, assume $v$ is contained in $S$. This implies that $v$ is also contained in the optimal selection~$M_i$ since the intersection of $M_i$ with the bag~$X_i$ must be exactly $S$. We claim that
\begin{align*}
\CiEntry = C_j(S\backslash\{v\},n_1',\dots, n_h') &+ p(v)\\
&  - \sum_{P_l: n_l' < |P_l|, n_l = |P_l|} p(P_l),
\end{align*}
where $n_l'$ is equal to $n_l$ if the node~$v$ is not adjacent to $P_l$ and equal to $n_l+1$ if $v$ and $P_l$ are adjacent.

Let $M_i$ be the optimal selection in $G_i$. If $v$ is removed from $M_i$, we know that $(M_i\backslash\{v\})\cap X_i = (M_i\backslash\{v\})\cap (X_j \cup\{v\})$. This is by definition equal to $S\backslash\{v\}$ and, thus, $M_i\backslash\{v\}$ is a feasible selection for $S\backslash\{v\}$ in $G_j$. The profit of $M_i\backslash\{v\}$ is less than the optimal profit~$C_j(S\backslash\{v\},n_1',\dots, n_h')$ and is given as the profit of $M_i$ in $G_i$ without the reward~$p(v)$ of the introduced node~$v$ and without the penalty of penalty sets~$P_l$ which are adjacent to $v$, i.e. without $\sum_{P_l: n_l' < |P_l|, n_l = |P_l|} p(P_l)$. Thus, we get 
\begin{align*}
\CiEntry \leq C_j(S\backslash\{v\},n_1',\dots, n_h') &+ p(v)\\
&  - \sum_{P_l: n_l' < |P_l|, n_l = |P_l|} p(P_l).
\end{align*}

Conversely, given a selection~$M_j$ with $M_j\cap X_j = S\backslash{v}$, one can extend this selection by adding $v$ to a selection~$M_i\coloneqq M_j\cup\{x\}$ in $X_i$ with $M_i\cap X_i = S$. Now the profit of $M_i$ is given as the sum of the profit of $M_j$ and the profit of the reward node~$v$ minus the penalty of the penalty sets~$P_l$ whose whole neighborhood lies in $X_i$ but not in $X_j$. Since the profit of $M_i$ is less than the profit~$\CiEntry$ of an optimal selection with corresponding degrees~$n_l$, we obtain

\begin{align*}
\CiEntry \geq C_j(S\backslash\{v\},n_1',\dots, n_h') &+ p(v)\\
&  - \sum_{P_l: n_l' < |P_l|, n_l = |P_l|} p(P_l).
\end{align*}
This settles the claim.

Again, the number of subsets of a bag is given by $2^{k+1}$ and the number of possible degree combinations is given by $n^{k+1}$. Thus, the running time of computing all values in an introduce node is 
\begin{align*}
\O(2^{k+1}\cdot n^{k+1}).
\end{align*}

\subsubsection{Join Nodes}
Let $X_i$ be a join node, i.e. $X_i = X_{j_1} = X_{j_2}$ for the two descendants~$ X_{j_1} = X_{j_2}$ of $X_i$ in the tree decomposition (see Figure~\ref{fig: schematic join node}).

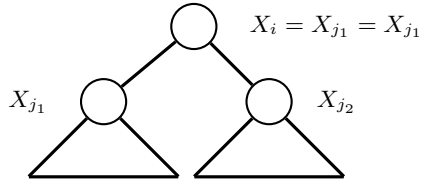
\begin{figure}[hbtp!]				
	\centering		
	\begin{tikzpicture}[every node/.style={fill=white,rectangle}, every edge/.style={draw=black,very thick}]
	\begin{scope}
	\draw
	(1.2,0) node (1) [bagvertex] {}
	(3.1,0) node (1a) [novertex] {\small{$X_i = X_{j_1} = X_{j_1}$}}
	(0,-1) node (2) [bagvertex] {}
	(-1,-1) node (2a) [novertex] {\small{$X_{j_1}$}}
	(-1,-2) node (3) [novertex] {}
	(1,-2) node (4) [novertex] {}		
	
	(2.2,-1) node (5) [bagvertex] {}
	(3.1,-1) node (5a) [novertex] {\small{$X_{j_2}$}}
	(1.2,-2) node (6) [novertex] {}
	(3.2,-2) node (7) [novertex] {};
	
	\path[-] (1) edge (2);
	\path[-] (2) edge (3);
	\path[-] (2) edge (4);
	\path[-] (3) edge (4);
	
	\path[-] (1) edge (5);
	\path[-] (5) edge (6);
	\path[-] (5) edge (7);
	\path[-] (6) edge (7);
	\end{scope}
	\end{tikzpicture}
	\caption{Schematic representation of a join node~$X_i$.}\label{fig: schematic join node}
\end{figure}

\begin{lemma}
	If there exists a penalty set node~$P_l$ in $G_{j_1}$ that has a neighbor~$v\in V_{i}$ in $G_i$ with $v\in V_{j_2}$ and $v\notin V_{j_1}$, then $P_l$ must be contained in the current bag~$X_i$. 
\end{lemma}
\begin{proof}Suppose for a contradiction, $P_l$ is not contained in the current bag~$X_i$. We know that there exists a bag $X_1$ that is a child of $X_{j_1}$ since $P_l$ is contained in $G_{j_1}$. Further, since $P_l$ is adjacent to the reward set node~$v$, there exists a bag~$X_2$ in which both are contained. Since $v$ is a node in $V_{j_2}$, the bag~$X_2$ is a child of $X_{j_2}$. Since the tree decomposition fulfills the path condition, the node~$P_l$ must be contained in every bag on the unique path from $X_1$ to $X_2$, in particular in bag~$X_i$.  
\end{proof}

\begin{lemma}\label{lemma: value of selection as sum of both bags}
	Let $M_{j_r}$ be a selection in $G_{j_r}$ for $r\in\{1,2\}$ for a given subset~$S\subseteq X_{j_1}=X_{j_2}$. Then, by taking the sum of the profit of both selections, the reward which is counted twice is exactly the reward that is obtained in~$S$.
\end{lemma}  
\begin{proof}
	Suppose $v$ is a reward set node that is contained in both the selections~$M_{j_1}$ and $M_{j_2}$. Since we have a nice tree decomposition, $v$ must also be contained in every bag on the unique path from $X_{v_1}$ to $X_{v_2}$, in particular also in $X_{j_1}$, $X_i$ and $X_{j_2}$. Thus, since $v$ is contained in $M_{j_r}$ and $M_{j_r}\cap X_{j_r} = S$ holds true for $r\in\{1,2\}$, also $v\in S$ holds true. This shows that the profit which was counted twice is exactly the profit gained by the reward set nodes~$S$.
\end{proof}

\begin{lemma}
	It holds:
	\begin{align*}
	C_i(S, n_1, \dots, n_h) = \max\{C_{j_1} & (S, n_1',\dots, n_h') + C_{j_2}(S, n_1'',\dots, n_h'') \\
	& - p(S) - \sum_{P_l: n_l', n_l'' < |P_l|, n_l = |P_l|} p(P_l)\\
	& \text{ for all } n_l', n_l'' \text{ with } n_l' + n_l'' = n_l \\
	& \text{ for all penalty set nodes } P_l \}
	\end{align*}
\end{lemma}
\begin{proof}
	Let $M_i$ be the optimal selection in $G_i$. Then, for $r=1,2$, $M_{j_r}\coloneqq M\cap X_{j_r}$ is a selection in $G_{j_r}$ since $M\cap X_{j_r} = S$ holds. Further, we know that $|M_{j_r} \cap P_l| = n_l^{(r)}$ holds for all penalty sets~$P_l$ and $r=1,2$. By Lemma~\ref{lemma: value of selection as sum of both bags}, the profit counted twice is given by $p(S)$. Also, we have to take into account the penalty of some penalty set nodes whose complete neighborhood is contained in $G_i$, but neither in $G_{j_1}$ nor in $G_{j_2}$. More formally, this penalty is given by the sum~$\sum_{P_l: n_l', n_l'' < |P_l|, n_l = |P_l|} p(P_l)$. Since the profit of $M_{j_r}$ is less than the optimal profit $C_{j_r}(S, n_1^{(r)}, \dots, n_h^{(r)})$, we obtain
	\begin{align*}
	C_i(S, n_1, \dots, n_h) \leq \max\{C_{j_1} & (S, n_1',\dots, n_h') + C_{j_2}(S, n_1'',\dots, n_h'') \\
	&  - p(S) - \sum_{P_l: n_l', n_l'' < |P_l|, n_l = |P_l|} p(P_l)\\
	& \text{ for all } n_l', n_l'' \text{ with } n_l' + n_l'' = n_l \\
	& \text{ for all penalty set nodes } P_l \}.
	\end{align*}
	
	Conversely, suppose $M_r$ is a selection in $G_{j_r}$ with $M\cap X_{j_r} = S$ for $r=1,2$. Let the corresponding profits be given by $C_{j_r}(S, n_1^{(r)},\dots, n_h^{(r)})$. Then, $M \coloneqq M_1 \cup M_2$ is a selection in $G_i$ with $M\cap X_i=S$ with profit $\CiEntry$ where $n_l' + n_l'' = n_l$ holds. Again, the profit counted twice is according to the lemma above given by $p(S)$ and the new penalty is given by the sum $\sum_{P_l: n_l', n_l'' < |P_l|, n_l = |P_l|} p(P_l)$. The profit of $M$ is less than the optimal profit $\CiEntry$ and this settles the proof. 
	%\todo{Aufschreiben!}
	%Let $M$ be a selection in $G_i$ with $M\cap X_i = S$ and profit~$\CiEntry$. Then, $M\cap X_{j_r}$ is a selection in $G_{j_r}$ with $M\cap X_{j_r} = S$ for $r=1,2$. The corresponding profits are given by $C_{j_1}(S, n_1',\dots, n_h')$ and $C_{j_2}(S, n_1'',\dots, n_h'')$ with $n_l' + n_l'' = n_l$. By the lemma above we know that the profit that is counted twice is exactly the profit obtained in $S$. Further, there might be some penalty set nodes whose complete neighbourhood is contained in $G_i$, but neither in $G_{j_1}$ nor in $G_{j_2}$. More formally, this is given by the sum~$\sum_{P_l: n_l', n_l'' < |P_l|, n_l = |P_l|} p(P_l)$. Thus, we get 
\end{proof}
Thus, by the same lines as above for the forget and introduce node, we can compute the entries of a join node in 
\begin{align*}
\O(2^{k+1}(n^{k+1})^2).
\end{align*}

By taking all the results from the different nodes together, we have proven the following theorem.
\begin{theorem}\label{thm: bounded tree width rpsp}
	Given an instance of the RPSP whose reduced connection graph has bounded tree-width~$k$ the RPSP can be solved in time $\O(n(2^{\O(k^3)}+2^{k+1}(n^{k+1})^2))$. $\hfill\Box$
\end{theorem}
%
%\subsection{Interpretation as a Hypergraph}
%Describe how to see this problem on a hypergraph. 

\subsection{Reward Singletons, Penalty Sets of Size 2}
In this subsection we discuss instances where all reward sets are singletons and all penalty sets are of size~2. As described above, this has a graph representation where the reward sets are represented by nodes and penalty sets by edges.

As seen in Theorem~\ref{thm: min-RPSP complexity}, the problem is $\NP$-complete even for uniform rewards~$a$ and uniform penalties~$b$ with $\frac{b}{a} > \frac{1}{\Delta}$, where $\Delta$ denotes the maximum degree of the given instance graph. 

For $\frac{b}{a} \geq 1$ the problem becomes polynomial solvable on chordal graphs since there always exists an optimal solution which is a maximum independent set which can be found in polynomial time. 

Unfortunately, the problem remains $\NP$-complete on instances that are represented by chordal graphs if we allow arbitrary rewards and penalties. 
\begin{theorem}
	The RPSP with reward singletons and penalty sets of size~2 is $\NP$-complete for arbitrary rewards and penalties on chordal graphs.
\end{theorem} 
\begin{proof}
	We show this by a reduction from MIS. Let a graph~$G$ be given. In order to construct a chordal instance graph~$G'$, we first add a penalty of $|V|+1$ to all existing edges. Now add all missing edges to the graph with a penalty of~$0$. Each node is associated with a reward of~$1$. Clearly, by this construction, we get a complete graph which is chordal. We now need to show that there exists a maximum independent set of size~$k$ in~$G$ if and only if there exists a selection~$S$ of nodes with profit at least~$k$ in~$G'$. 
	
	Suppose we are given a maximum independent set of size~$k$ in $G$. We need to show that this selection~$S$ of nodes has a profit of at least~$k$. Since we are given an independent set, no two nodes that are incident to an existing edge are in the selection~$S$. Thus, the profit of~$S$ is given by the number of nodes, i.e. the size of the independent set. Conversely, suppose we are given a selection~$S$ with profit at least~$k$ in $G$. Since $k > 0$, there cannot be two nodes contained in $S$ that are incident to an existing edge. Thus, the selection~$S$ induces an independent set of size at least $k$ in $G$.  
\end{proof} 

\section{Application as a Graph Theoretical Problem}\label{sec: rpsp: graph}
The \emph{subgraph selection problem (SGSP)} is defined as follows. Let a graph~$G$ be given. Furthermore, let $\mathcal{R}$ and $\mathcal{P}$ be two sets of subgraphs of $G$ with associated weights~$w:\mathcal{R}\cup \mathcal{P}:\mapsto \R_+$. The task consists in finding a selection~$S$ of nodes such that
\begin{align*}
\sum_{R\in\mathcal{R}: R\cap G|_S \neq \emptyset} w(R) - \sum_{P\in\mathcal{P} : P\subseteq G|_S} w(P)
\end{align*}
is maximized. The SGSP can be seen as a generalization of the \emph{graph coverage problem}. If we set $\mathcal{R}$ as the set of all single nodes with $w|_\mathcal{R} \equiv 1$ and $\mathcal{P}$ as the set of all paths of length $1$ with $w|_P \equiv |V| + 1$, the SGSP can be seen as the MIS problem. Thus, a reduction from MIS can be easily used to show that the general SGSP is $\NP$-hard.

\subsection{The SGSP on Trees}
In this subsection we consider the SGSP on trees. We use the ideas presented in \cite{krumke2002budgeted, hunt2002parallel, stearns1996algebraic} in order to show that the SGSP can be solved in polynomial time on trees if the reward subgraphs are given as nodes and the penalty subgraphs are given as connected subgraphs such that each node is only contained in at most a constant number of penalty subgraphs. First, we show that the SGSP is, in general, hard to solve, even on star graphs.

\begin{theorem}[Complexity of SGSP on trees]\label{thm: compl sgsp tree}
	Given a star graph~$G$, reward subgraphs are single nodes and penalty subgraphs are paths of length at most three. Then the SGSP on $G$ is $\NP$-hard.
\end{theorem}
\begin{proof}
	We prove this by a reduction from MIS. Given an instance of the MIS on the graph~$G'$, we construct a graph~$G$ in the following way. For the node set~$V(G)$ we take the node set~$V(G')$ together with an artificial node~$c$. Further, we connect all nodes in $V(G)$ to $c$. Thus, the resulting graph is a star. We take the reward subgraphs as singletons, i.e. each node~$v$ in $V(G)$ defines a reward subgraph with reward equal to $1$ if $v\in V(G')$ or reward~$M>|V(G')|$ if $v=c$. Further, for each edge~$(u,v)\in E(G')$, we add a penalty subgraph as the unique path~$(u,c,v)$ with penalty~$|V|+1$. 
	
	There exists a solution to the MIS on $G'$ of size~$k$ if and only if there exists a solution to the SGSP with value~$M+k$. Given an independent set~$S$ of size~$k$, we choose in $G$ all the nodes from $S$ together with the node~$c$. Since $S$ is an independent set, no penalty subgraph is chosen completely. Thus, we obtain a solution of the SGSP with value~$M+k$. Conversely, given a solution of the SGSP with value~$M+k$, we know that the node~$c$ has to be chosen. Since $k< |V|+1$, no penalty subgraph can be chosen completely. Thus, no pair of nodes of the induced node set in $G'$ is adjacent and we obtain an independent set of size~$k$. 
	%	We prove this by a reduction from the Hitting Set Problem (HSP). An instance of the HSP is given by a ground set~$S$ with $m$ elements, a collection~$C=\{S_1,\dots,S_n\}$ of subsets of $S$ and an integer~$k\in\N$. The HSP now asks if there exists a subset~$S'\subseteq S$ with $|S'|\leq k$ such that $S_i\cap S' \neq \emptyset$ for all $i=1,\dots,n$. We construct an instance of the SGSP as follows. For each element in $S$ we introduce a node and add an extra \emph{center} node which is adjacent to all other nodes. Now, we say that the reward subgraphs are given by the subgraphs induced from the sets~$S_1,\dots,S_n$ together with the center node. The penalty subgraphs are given the induced subgraphs from the subsets of $S$ with exactly $k+1$ elements together with the center node. Furthermore, we assume the uniform weight for reward subgraphs is $1$ and the uniform weight for penalty subgraphs is $|V|+1$. 
	%	
	%	Given a solution~$S'$ of the HSP of size~$k$, we obtain directly a solution of the SGSP with objective value~$n$, since from each reward subgraph at least one node is chosen and no penalty subgraph is chosen completely. Conversely, given a solution of the SGSP with objective value~$n$, we know that no penalty subgraph can be chosen completely. Since each reward subgraph gives at most $1$ reward, we know that from each reward subgraph at least one node has to be chosen. This completes the proof. 
\end{proof}

By this reduction, we know that the SGSP is $\NP$-hard even on instances that are given by a star graph. Therefore, we restrict ourselves to the case where we assume that all reward subgraphs are given by nodes in $V(G)$, all penalty subgraphs are connected and each node~$v$ is only contained in a bounded number of penalty sets. Note that without the frequency restriction by Theorem~\ref{thm: compl sgsp tree} the problem is still hard. We follow the notation of \cite{krumke2002budgeted} to introduce the \emph{frequency}. %TODO satz without frequency....
For a node~$v\in V$, we denote its frequency by $\Phi_v \coloneqq \{P\in\mathcal{P} : v\in V(P)\}$, which is the number of penalty subgraphs containing the node~$v$. The \emph{maximum frequency}~$\Phi$ is then defined as
\begin{align}
\Phi \coloneqq \max_{v\in V} \Phi_v.
\end{align}
Note that the SGSP can be formulated by the IP~\eqref{ip: min-RPSP}. The number of variables of the formulation can be reduced by identifying the node variables~$x$ with the reward subgraph variables~$y$. This can be done since one can add dummy reward subgraphs with reward~$0$ for all nodes that are not contained in an already existing reward subgraph. Thus, we get the following integer program.
\begin{align}
\textbf{(SGSP)} \qquad \max \quad &\sum_{R=\{v\}\in\mathcal{R}}w(R)x_v - \sum_{P\in\mathcal{P}} w(P)z_j  \label{ip: sgsp}\\
\text{s.t.}\quad & \sum_{u\in P} (x_u - 1) + 1 \leq z_j \qquad \qquad  \text{for all } P\in\mathcal{P} \label{cons: sgsp}\\
& x_u,z_j \in \{0,1\}  \label{var: decision sgsp}
\end{align}

Given such a formulation, let $X$ be a set of variables and $C$ be a set of constraints on~$X$. Then the \emph{constraint graph} $BP(X,C)$ associated with $(X,C)$ is defined as the bipartite graph with node classes~$X$ and $C$, and edges between $x\in X$ and $c\in C$ if and only if the variable~$x$ appears in constraint~$c$. The \emph{interaction graph} $IG(X,C)$ for $(X,Z)$ is the graph with node set~$X$ and edges between $x_1,x_2\in X$ if and only if they have a common neighbor in the constraint graph~$BP(X,C)$.

%The \emph{interaction graph} $IG(X,C)$ for $(X,Z)$ is the graph with node set~$X$ and edges between $x_1,x_2\in X$ if and only if they have a common neighbour in the constraint graph~$B(X,C)$. Given the IP formulation, we have three kind of variables, the node variables~$x$, the reward subgraph variables~$y$ and the penalty subgraph variables~$z$. Since two penalty subgraph variables are not contained together in any constraint, the edges~$(z_{j_1}, z_{j_2})$ do not exist in the interaction graph. The same holds for edges between reward subgraph nodes and between a reward subgraph node and a penalty subgraph node. Two node variables are contained in the same constraint, if and only if there exists a subgraph, either a reward or a penalty subgraph, which contains both. A node variable node is adjacent to a reward subgraph node if and only if the node is contained in the reward subgraph -- analogously the same holds for penalty subgraphs. Thus, the edges in the interaction graph are of the form
%\begin{align*}
%e = \begin{cases}
%(x_u, x_v) & \text{ if both are contained in a set together. }\\
%(x_u, y_i) & \text{ if } u\in A_i \\
%(x_u, z_j) & \text{ if } u\in B_j.
%\end{cases}
%\end{align*}
%
%This shows that the interaction graph is completely defined by the node variables and how they are contained in reward and penalty subgraphs. By the construction above we see that all nodes of the same reward subgraph form a clique in the interaction graph. Thus, the size of the sets are a lower bound on the tree width. 

%We want to make use of the following two results.
Recall the following two results by Stearns et al. and Hunt et al. which provide a polynomial running time for solving integer programs with an interaction graph that has bounded tree-width. 

\begin{theorem}[cf. \cite{stearns1996algebraic}]\label{thm: solving interaction graph}		
	Let~$p$ be a polynomial. Let $Z$ be a set of variables taking values from the domain $\{0, \dots , K\}$, where $K \in \O(p(|Z|))$, and let $C$ be a set of constraints on $Z$. Then for any fixed $k \in \N$ and any non-negative vector~$c=(c_z)_{z\in Z } \in \{0,\dots, p(|Z|)\}^Z$, the integer program of maximizing~$\sum_{z\in Z} c_z \cdot z$ subject to the constraints~$C$, restricted to those instances where the interaction graph for $(Z, C)$ has bounded tree-width at most~$k$, can be solved in time $K^{\O(k)}$.
\end{theorem}

\begin{theorem}[cf. \cite{hunt2002parallel}]\label{thm: bounded constraint graph}
	Let $Z$ be a set of variables and $C$ be a set of constraints on $Z$. Suppose that each	constraint contains at most $k$ variables. Let $BP(Z, C)$ be the constraint graph associated with $(Z, C)$ and $IG(Z, C)$ the corresponding interaction graph. Then,
	\begin{align*}
	tw(IG(Z, C)) \in \O(k \cdot tw(BP(Z, C)).
	\end{align*}
\end{theorem}
In order to apply these two results to the SGSP, we need to show that the constraint graph of an instance has bounded tree-width. This is done in the following lemma.  

\begin{lemma}\label{lem: sgsp interaction graph}
	Let the SGSP on a tree~$G$ be restricted to the case where the reward subgraphs are singletons, the penalty subgraphs are connected and the maximum frequency~$\Phi$ is bounded by a constant~$B$. Then the constraint graph of \eqref{ip: sgsp} with node set introduced by the variables~$x,z$ and constraints~\eqref{cons: sgsp} has a tree-width of at most $\Phi$.
\end{lemma}
\begin{proof}
	We prove this by constructing a tree decomposition~$(\mathcal{X},T)$ of the constraint graph of problem~\eqref{ip: sgsp} - \eqref{var: decision sgsp}. For each node~$v\in V(G)$ we introduce a bag
	$X_v\coloneqq \{x_v, C_1,\dots, C_{\Phi_v}\}$, where the $C_i$ are the constraints such that $v$ is contained in the corresponding penalty subgraph~$P_i$. Further, we introduce bags~$X_p \coloneqq \{c_p, z_p\}$. The decomposition tree~$T$ has the same structure as the graph~$G$ where we add additional leaves~$i$ such that $X_i$ is of the form $\{C_i, z_i\}$ and adjacent to an arbitrary bag that contains~$C_i$. The biggest size of a bag is now $\Phi + 1$ and thus we get a tree-width of $\Phi$. 
	
	By construction, it holds $\cup_{i\in V(T)} X_i = V(G)$. Also, each pair of adjacent nodes are in at least one shared bag. It remains to show that for each node~$v$ in the constraint graph, the set of nodes $\{i: v\in X_i\}$ forms a subtree of $T$. If $v$ is a node corresponding to a node variable~$x_v$, it is only contained in one bag. Similar, also by construction, a penalty subgraph node is only contained in one bag. Now let $v$ be a constraint node. Suppose the set~$\{X_i: v\in X_i\}$ does not form a subtree of $T$. Since $T$ is a tree, the only possibility is that the set induces a not connected subgraph. This would imply that there exists a penalty subgraph~$P$ that is not connected -- a contradiction to the assumption. 	
\end{proof}

As we have seen, the SGSP on trees has, under the restrictions of Lemma~\ref{lem: sgsp interaction graph}, an interaction graph with bounded tree-width. By applying Theorem~\ref{thm: bounded constraint graph} and Theorem~\ref{thm: solving interaction graph}, we know that the SGSP on trees under said conditions can be solved in polynomial time. We summarize this in the following theorem. 

\begin{theorem}[Complexity of SGSP on trees]
	Let $c\in\N$ be an integer. Furthermore, let the SGSP on trees be restricted such that maximum frequency bounded by $c$. Also, let the reward subgraphs be given as singletons, the penalty subgraphs be connected and the weight function be polynomially bounded. Then the SGSP is solvable in polynomial time. $\hfill\Box$
\end{theorem}

We know that the independent set problem on graphs is hard to solve, even if we restrict ourselves to graphs with bounded degree. The result from above is no contradiction since in the reduction graph of the proof of Theorem~\ref{thm: compl sgsp tree} the property of bounded frequency induces a bound on all nodes, except from the "center" node which is still connected to all other nodes.

\section{Conclusion}
In this paper we discussed the RPSP, a combinatorial optimization problem which can be viewed as a combination of the SCP and the HSP. We gave complexity results for the general minimization problem as well as for the the maximization problem. While it turns out that the first one is in general hard to solve, the latter one is solvable in polynomial time. Furthermore, we gave a formulation of the RPSP as an integer program. A short numerical study shows that a rounding approach depending on the linearization of the integer program yields in good solutions, since the average approximation ratio for all the tested instances is greater than 0.95. 

We considered problem variants of the minimization RPSP. For the laminar RPSP we obtained a polynomial time algorithm depending on a tree representation of the problem instance and a flow computation on the corresponding network graph. If the reward sets are given as singletons and the graph depending on the instance has bounded treewidth, we obtained a polynomial time algorithm based on a dynamic programming approach. If the reward sets are given as singletons, penalty sets are of size exactly two and uniform weights are given, we showed that one can compute a solution on chordal instance graphs by using the fact that the maximum independent set problem is solvable in polynomial time on chordal graphs. Unfortunately, the problem remains hard to solve when considering arbitrary weights. 

Furthermore, we gave a generalization of the RPSP as a combinatorial problem from a graph theoretical point of view. In this problem, one tries to find a selection of nodes such that some desired subgraphs are covered while others are avoided. Using a reduction from the MIS, it turns out that the SGSP is in general hard to solve, even on trees and under additional size restrictions of the subgraphs. By using results from \cite{hunt2002parallel, stearns1996algebraic, krumke2002budgeted}, the SGSP can be solved in polynomial time if the maximum frequency of the instance graph is bounded. 

The SGSP raises many research question worth to address. While the question of the complexity is settled, the problem of finding an approximation of the SGSP remains open. Furthermore, since the proof of the complexity is based on a reduction from the MIS, it might be worth to consider instances where the MIS is solvable in polynomial time, such as chordal graphs.   

\bibliographystyle{plain}
\bibliography{references}
\end{document}